\documentclass[graybox]{svmult}
\usepackage{mathptmx}       
\usepackage{newtxmath}

\usepackage{helvet}         
\usepackage{courier}        
\usepackage{type1cm}        
\usepackage{makeidx}         
\usepackage{graphicx}        
\usepackage{multicol}        
\usepackage[bottom]{footmisc}
\usepackage{tikz} 

\usepackage{url}

\usepackage{amsfonts}
\usepackage{amsmath}
\usepackage{amssymb}

\usepackage{braket}

\IfPackageLoadedTF{mathptmx}{
\newcommand{\LLhb}{\mkern1.5mu\mathchar'26\mkern-7.25mu h} 
}{
\newcommand{\LLhb}{\hbar} 
}%
\newcommand		{\LLhd}		{h^d}

\newcommand		{\LLNN}		{\mathbb N}			
\newcommand		{\LLRR}		{\mathbb R}			
\newcommand		{\LLRd}		{\LLRR^d}
\newcommand		{\LLRdd}	{\LLRR^{2d}}
\newcommand		{\LLL}		{\mathcal L}		

\newcommand		{\com}[1]		{\left[{#1}\right]}		
\newcommand		{\LLn}[1]			{\left|{#1}\right|}
\newcommand		{\LLnrm}[1]		{\left\|{#1}\right\|}
\newcommand		{\LLsnrm}[1]		{\lVert #1\rVert}
\newcommand		{\LLbnrm}[1]		{\big\lVert #1\big\rVert}

\newcommand		{\LLNrm}[2]		{\LLnrm{#1}_{#2}}
\newcommand		{\LLsNrm}[2]		{\LLsnrm{#1}_{#2}}
\newcommand		{\LLbNrm}[2]		{\LLbnrm{#1}_{#2}}

\newcommand		{\LLd}		{\mathop{}\!\mathrm{d}}		
\DeclareMathOperator{\LLtr}		{Tr}			

\newcommand		{\LLTr}[1]		{\LLtr\!\left( #1 \right)} 	

\newcommand		{\LLintd}			{\int_{\LLRd}}
\newcommand		{\LLintdd}		{\int_{\LLRdd}}

\newcommand		{\LLcC}			{\mathcal{C}}

\newcommand		{\LLop}		{{\boldsymbol{\rho}}}	
\newcommand		{\LLopp}		{{\boldsymbol{p}}}

\newcommand		{\LLDh}		{\pmb{\nabla}}	
\newcommand		{\LLDhx}[1]	{\LLDh_{\!x} #1}			
\newcommand		{\LLDhv}[1]	{\LLDh_{\!v} #1}		

\begin{document}

\title*{Uncertainty Relation for the Wigner--Yanase Skew Information and Quantum Sobolev Inequalities}
\titlerunning{Uncertainty relation for the Wigner--Yanase Skew information} 
\author{Laurent Lafleche} 
\institute{Laurent Lafleche \at Unité de Mathématiques pures et appliquées, École Normale Supérieure de Lyon, 69364 Lyon, France \at \email{laurent.lafleche@ens-lyon.fr} 
}

\maketitle

\abstract{
	This note explores uncertainty inequalities for quantum analogues of the Fisher information including the Wigner--Yanase skew information, and their connection to the quantum Sobolev inequalities proved by the author in (\textit{Journal of Functional Analysis}, 286 (10) 2024). Some additional inequalities concerning commutators are derived and others are left as open problems.
}

\keywords{Operator inequalities; trace inequalities; Sobolev inequalities; semiclassical approximation.}
\\
{{\bf MSC2020:} 81S07; 46E35; 81S30 (47A30).}

\section{Quantum Skew Information and Variance}\label{sec:skew}

	Let $\LLop$ be a density operator, that is, a positive compact operator, acting on the space $L^2(\LLRd)$ of square-integrable functions, with $\LLTr{\LLop}=1$. The \textit{quantum variance} of an observable $K$ with respect to $\LLop$ is given by
	\begin{equation*}
		\sigma_K(\LLop)^2 = \LLtr(\LLn{K}^2\LLop) - \LLn{\LLTr{K\,\LLop}}^2 ,
	\end{equation*}
	where $\LLn{K}^2 = K^*K$. It satisfies the famous Heisenberg(--Robertson) uncertainty inequality
	\begin{equation*}
		\sigma_A(\LLop) \, \sigma_B(\LLop) \geq \frac{1}{2} \LLn{\LLTr{\com{A,B}\LLop}}
	\end{equation*}
	which reduces to $\sigma_x(\LLop) \, \sigma_\LLopp(\LLop) \geq \frac{d\LLhb}{2}$ if $x$ is the operator of multiplication by $x\in\LLRd$ and $\LLopp = -i\LLhb \nabla$. Here $\com{A,B} = A\cdot B - B\cdot A$ is the commutator of $A$ and $B$.

	A quantum analogue of the Fisher information is the \textit{Wigner--Yanase skew information}~\cite{wigner_information_1963} with respect to a self-adjoint operator $K$, which is defined by
	\begin{equation}\label{eq:skew}
		I_K(\LLop) = \frac{1}{2} \LLTr{\LLn{\com{K,\sqrt{\LLop}}}^2}.
	\end{equation}
	Another quantum analogue of the Fisher information, which makes it possible to obtain a quantum analogue of the Cramér--Rao inequality~\cite{helstrom_minimum_1967} when the parameter change is implemented through a Hamiltonian dynamics, is the \textit{symmetric logarithmic derivative Fisher information} (see, e.g.,~\cite{holevo_probabilistic_2011, braunstein_statistical_1994, hayashi_two_2002}) defined by
	\begin{equation}\label{eq:skew_J}
		J_K(\LLop) = \LLTr{\LLn{L_K(\LLop)}^2\LLop} \quad \text{ with } \quad \frac{1}{2}\left(L_K(\LLop)\LLop + \LLop L_K(\LLop)\right) = \frac{1}{i}\com{K,\LLop} .
	\end{equation}
	If $\LLop$ is a compact positive operator which can be diagonalized in the form
	\begin{equation*}
		\LLop = \sum_{j\in J} \lambda_j \,\ket{\psi_j}\bra{\psi_j}
	\end{equation*}
	for some set $J\subset \LLNN$, where $(\psi_j)_{j\in J}$ are the eigenvectors of $\LLop$ and $(\lambda_j)_{j\in J}$ the corresponding eigenvalues, then the above quantities can be written
	\begin{equation*}
		I_K(\LLop) = \frac{1}{2}\sum_{j\in J} \Big|\sqrt{\lambda_j}-\sqrt{\lambda_k}\Big|^2 \LLn{\braket{\psi_j|K\psi_k}}^2
	\end{equation*}
	and
	\begin{equation*}
		J_K(\LLop) = 4\sum_{j,k\in J} \frac{\LLn{\lambda_j-\lambda_k}^2}{\lambda_j+\lambda_k} \LLn{\braket{\psi_j|K\psi_k}}^2
	\end{equation*}
	where the sum is over all the $j$ and $k$ such that $\lambda_j+\lambda_k\neq 0$.
	
	These quantum informations are smaller than the variance, and more precisely
	\begin{equation*}
		\sigma_K^2(\LLop) \geq \frac{1}{4} \,J_K(\LLop) \geq I_K(\LLop) \, ,
	\end{equation*}
	with equality when $\LLop$ is a pure state, i.e. a rank-one projection. As proved in~\cite{kosaki_matrix_2005}, the Heisenberg-like uncertainty inequality $\sqrt{I_A(\LLop)\,I_B(\LLop)} \geq \frac{1}{2} \LLn{\LLTr{\com{A,B}\LLop}}$ does not hold for the Wigner--Yanase skew information, as was wrongly claimed in~\cite{luo_informational_2004} and later corrected in \cite{luo_correction_2005}. Uncertainty principles for modifications of the skew information were investigated in \cite{luo_heisenberg_2005, yanagi_wigneryanasedyson_2010, yang_generalized_2022}. In the case of the symmetric logarithmic derivative Fisher information~\eqref{eq:skew_J}, the analogue of the Cramér--Rao inequality follows directly from the Cauchy--Schwarz inequality for the trace~\cite{helstrom_minimum_1967, frowis_tighter_2015} and can be viewed as an improved Heisenberg uncertainty principle, since it reads
	\begin{equation*}
		\sigma_A(\LLop)\, \sqrt{\tfrac{1}{4}J_B(\LLop)} \geq \frac{1}{2}\LLn{\LLTr{\com{A,B}\LLop}} .
	\end{equation*}
	
	The problem of finding a lower bound to $I_A(\LLop)\,I_B(\LLop)$ was solved recently by the author in~\cite{lafleche_quantum_2024} in the case $\left(A,B\right) = \left(x,\LLopp\right) = \left(x,-i\LLhb \nabla\right)$ in dimension $d\geq 2$, and involves higher-order Schatten norms instead of the trace. Recalling the definition of Schatten norms of order $p\in[1,\infty)$
	\begin{equation*}
		\LLNrm{\LLop}{p} = \left(\LLTr{\LLn{\LLop}^p}\right)^{1/p}
	\end{equation*}
	and setting $\LLNrm{\LLop}{\infty}$ to be the operator norm, the uncertainty inequality for the Wigner–Yanase information takes the following form.
	\begin{theorem}[\cite{lafleche_quantum_2024}, Corollary~2.2]\label{thm:dim_general}%
		Let $d\geq 2$ and $\LLop$ be a compact operator on $L^2(\LLRd)$. Then
		\begin{equation*}
			\frac{1}{4}\sqrt{J_x(\LLop)\,J_\LLopp(\LLop)} \geq \sqrt{I_x(\LLop)\,I_\LLopp(\LLop)} \geq \frac{\LLhb}{8\pi C_d} \LLNrm{\LLop}{\frac{d}{d-1}}
		\end{equation*}
		where $\LLcC^\textnormal{S}_{1,2} \leq C_d \leq \LLcC^\textnormal{S}_{1,2} + \frac{1}{\sqrt{8\pi}}$ with $\LLcC^\textnormal{S}_{1,2}$ the optimal constant in the Sobolev inequality $\dot H^1(\LLRdd)\to L^p(\LLRdd)$, i.e. $(\LLcC^\textnormal{S}_{1,2})^2 = \left(4\pi\right)^{-1/(2d)} \frac{\Gamma(d+1/2)^{1/d}}{d\left(d-1\right)\pi}$.
	\end{theorem}
	
	In dimension $d=1$, one can obtain a slightly weaker result assuming that $\LLop$ is a density operator.
	\begin{theorem}\label{thm:dim_1}
		Let $\LLop$ be a density operator on $L^2(\LLRR)$. Then for any $p\in(1,\infty)$,
		\begin{equation*}
			\frac{1}{4}\sqrt{J_x(\LLop)\,J_\LLopp(\LLop)} \geq \sqrt{I_x(\LLop)\,I_\LLopp(\LLop)} \geq \frac{\LLhb}{8\pi C_{1/p'}^{2p'}} \LLNrm{\LLop}{p}^{p'}
		\end{equation*}
		with $C_s = \frac{1}{\left(8\pi\right)^{s/2}} + \frac{\pi^{\left(1-s\right)/2}}{2^s\sin(\pi s)^{1/2}\Gamma(s)\sqrt{s}}$.
	\end{theorem}
	
\section{Link with Quantum Sobolev Inequalities}\label{sec:sobolev}
	
	The proof of Theorem~\ref{thm:dim_general} is based on semiclassical approximation techniques, since it is a special case of the quantum Sobolev inequalities derived in~\cite{lafleche_quantum_2024}. We also refer to~\cite{ruzhansky_sobolev_2025}, where these inequalities were recently proved in the closely related setting of noncommutative Euclidean spaces.
	
	According to the correspondence principle of quantum mechanics, the following Table~\ref{tab:dict} of analogies allows one to pass from quantum to classical mechanics.
	\begin{table}[h] 
		\caption{The quantum--classical dictionary}\label{tab:dict}
		\resizebox{\textwidth}{!}{%
		\begin{tabular}{ccccc}
			\hline\noalign{\smallskip}
			&\hspace{5pt} & Classical &\hspace{5pt}& Quantum  \\
			\noalign{\smallskip}\svhline\noalign{\smallskip}
			Phase space variables  &&  $(x,v)\in\LLRdd$ && $(x,\LLopp) = (x,-i\LLhb\nabla)$
			\\
			Phase space density  && Probability distribution $f(x,v)$ && density operator $\LLop$
			\\
			Expectation && $\LLintdd f(x,v) \LLd x\LLd v$ && $\LLhd \LLTr{\LLop}$
			\\
			Lie Brackets && $\{g,f\} = \nabla_x g\cdot\nabla_v f - \nabla_v g\cdot\nabla_x f$ && Commutator $\frac{1}{i\LLhb}\com{A,B} = \frac{1}{i\LLhb}\left(AB - BA\right)$
			\\
			\noalign{\smallskip}\hline\noalign{\smallskip}
		\end{tabular}}
	\end{table}
	In particular, $\LLDhx{\LLop} := \frac{1}{i\LLhb}\com{x,\LLop}$ corresponds to the quantum analogue of $\nabla_x f$ and $\LLDhv{\LLop} := \com{\nabla,\LLop}$ is the quantum analogue of $\nabla_v f$. Therefore, the analogues of the skew information $I_x(\LLop)$ and $I_\LLopp(\LLop)$ are, up to multiplicative constants, $\LLbNrm{\nabla_v \sqrt{f}}{L^2(\LLRdd)}^2$ and $\LLbNrm{\nabla_x \sqrt{f}}{L^2(\LLRdd)}^2$, which correspond to the classical Fisher information with respect to $x$ and $v$. On the other hand, classical Sobolev inequalities tell us that for $d>1$,
	\begin{equation*}
		\LLNrm{g}{L^\frac{2d}{d-1}(\LLRdd)} \leq \LLcC^\textnormal{S}_{1,2} \LLNrm{g}{\dot H^1(\LLRdd)} = \LLcC^\textnormal{S}_{1,2} \, \Big(\LLNrm{\nabla_x g}{L^2(\LLRdd)}^2 + \LLNrm{\nabla_v g}{L^2(\LLRdd)}^2\Big)^{1/2}.
	\end{equation*}
	Taking $g(x,v) = \sqrt{f(\lambda x,v/\lambda)}$ in the above inequality, optimizing with respect to $\lambda$ and squaring the inequality gives
	\begin{equation*}
		2\, (\LLcC^\textnormal{S}_{1,2})^2 \, \LLbNrm{\nabla_x \sqrt{f}}{L^2(\LLRdd)} \, \LLbNrm{\nabla_v \sqrt{f}}{L^2(\LLRdd)} \geq \LLbNrm{f}{L^\frac{d}{d-1}} .
	\end{equation*}
	
	In the quantum setting, the analogue of Sobolev inequalities read for any compact operator $\LLop$ in $L^2(\LLRd)$
	\begin{equation*}
		\LLNrm{\LLop}{\LLL^q} \leq C_{d,s,p} \LLNrm{\LLDh\LLop}{\LLL^p}
	\end{equation*}
	where the scaled Schatten norms $\LLNrm{\cdot}{\LLL^p}$ and the quantum gradient $\LLDh$ were defined according to the dictionary by
	\begin{equation*}
		\LLNrm{\LLop}{\LLL^p} := h^{d/p} \LLNrm{\LLop}{p} \quad \text{ and } \quad \LLDh\LLop := \left(\LLDhx{\LLop},\LLDhv{\LLop}\right) ,
	\end{equation*}
	so that $\LLn{\LLDh\LLop}^2 = \LLn{\LLDhx{\LLop}}^2+\LLn{\LLDhv{\LLop}}^2$. The case of fractional Sobolev spaces is treated as well in~\cite{lafleche_quantum_2024}.
	
	A useful tool for the proof is the Wigner transform, which is the inverse of the Weyl quantization, and is defined in terms of the integral kernel $\LLop(x,y)$ of the operator $\LLop$ by
	\begin{equation*}
		f_\LLop(x,v) = \LLintd e^{-iy\cdot v/\LLhb} \, \LLop(x+\tfrac{y}{2}, x-\tfrac{y}{2}) \LLd y \, .
	\end{equation*}
	It is an isometry from $\LLL^2$ to $L^2(\LLRdd)$ and satisfies $\nabla f_\LLop = f_{\LLDh \LLop}$.
	
	\begin{proof}[Proof of Theorem~\ref{thm:dim_1}]
		When $d=1$, the H\"older inequality for the trace and the fractional quantum Sobolev inequality~\cite{lafleche_quantum_2024} together with the fact that the Wigner transform is an isometry from $\LLL^2$ to $L^2(\LLRR^2)$ tell us that for any $q\in[2,\infty)$,
		\begin{equation*}
			\LLNrm{\LLop}{\LLL^q} \leq C_s \LLNrm{f_\LLop}{\dot H^s(\LLRR^2)}
		\end{equation*}
		with $s = 1-2/q$ and $C_s = \frac{1}{\left(8\pi\right)^{s/2}} + \frac{1}{2^s\,\pi^{s/2}} \frac{\Gamma(1-s)^{1/2}}{\Gamma(1+s)^{1/2}} =  \frac{1}{\left(8\pi\right)^{s/2}} + \frac{\pi^{\left(1-s\right)/2}}{2^s\sin(\pi s)^{1/2}\Gamma(s)\sqrt{s}}$ using the reflection formula for the Gamma function. By interpolation of $\dot H^s$ spaces, or, in a more elementary way, using the Fourier transform and H\"older's inequality, it gives
		\begin{equation*}
			C_s^{-1}\LLNrm{\LLop}{\LLL^q} \leq \LLNrm{f_\LLop}{\dot H^1(\LLRR^2)}^s \LLNrm{f_\LLop}{L^2(\LLRR^2)}^{1-s} = \left(\LLNrm{\LLDhx\LLop}{\LLL^2}^2 + \LLNrm{\LLDhv\LLop}{\LLL^2}^2\right)^\frac{s}{2} \LLNrm{\LLop}{\LLL^2}^{2/q} .
		\end{equation*}
		Applying this inequality to $\sqrt{\LLop}(\lambda x,\lambda y)$ and optimizing with respect to $\lambda$ leads to
		\begin{equation*}
			C_s^{-2}\LLNrm{\LLop}{\LLL^{q/2}} \leq 2^s \LLNrm{\LLDhx\sqrt{\LLop}}{\LLL^2}^s \LLNrm{\LLDhv\sqrt{\LLop}}{\LLL^2}^s \LLNrm{\LLop}{\LLL^1}^{2/q} 
		\end{equation*}
		which gives with $p = q/2$
		\begin{equation*}
			\frac{\LLhb}{4\pi C_{1/p'}^{2p'}} \LLNrm{\LLop}{p}^{p'} \leq \LLNrm{\com{\LLopp,\sqrt{\LLop}}}{2} \LLNrm{\com{x,\sqrt{\LLop}}}{2} \LLNrm{\LLop}{1}^{1/(p-1)} .
		\end{equation*}
		This gives the result by definition of the Wigner--Yanase skew information and since $\LLNrm{\LLop}{1} = 1$.
	\end{proof}

\section{Open Questions}\label{sec:quest}

	We give here some remaining open questions concerning the Wigner--Yanase skew information and quantum Sobolev spaces.
	
\subsection{Uncertainty Inequality}\label{sec:quest_uncertainty}

	Concerning the uncertainty inequality for the Wigner--Yanase skew information, the above results only address the case $A=x$ and $B = \LLopp$. A natural question is whether more general inequalities can be obtained for two non-commuting operators $A$ and $B$. In dimension $d=1$, in the classical commutative case, if $z\mapsto \left(\alpha(z),\beta(z)\right)$ is a $C^1$-diffeomorphism and $f\in L^1(\LLRR^2)$ is a probability density such that $f\circ \Phi^{-1}$ converges to $0$ at infinity,
	\begin{equation}\label{eq:classical_uncertainty_1d}
		\frac{1}{2\,(\LLcC^\textnormal{S}_{1,1})^2} \LLNrm{\LLn{\{\alpha,\beta\}}^\frac{1}{2}f}{L^2(\LLRR^2)}^2 \leq \LLNrm{\{\alpha,\sqrt{f}\}}{L^2(\LLRR^2)} \LLNrm{\{\beta,\sqrt{f}\}}{L^2(\LLRR^2)} ,
	\end{equation}
	where the brackets denote the Poisson brackets.	Hence, in the operator setting with $d=1$, one would expect the existence of a constant $C>0$ such that for any operators $A$ and $B$ associated with functions $\alpha$ and $\beta$ inducing a $C^1$-diffeomorphism, and any density operators $\LLop$ with some condition of compatibility with $A$ and $B$,
	\begin{equation*}
		C \LLNrm{\LLn{\com{A,B}}^\frac{1}{2}\LLop}{\LLL^2}^2 \leq \LLNrm{\com{A,\sqrt{\LLop}}}{\LLL^2} \LLNrm{\com{B,\sqrt{\LLop}}}{\LLL^2} .
	\end{equation*}
	That is, one would expect that for a rather large class of operators,
	\begin{equation*}	
		\frac{1}{4}\sqrt{J_A(\LLop)\,J_B(\LLop)} \geq \sqrt{I_A(\LLop)\,I_B(\LLop)} \geq C \, \LLTr{\LLn{\com{A,B}}\LLop^2}\, . 
	\end{equation*}
	
	\begin{proof}[Proof of Inequality~\eqref{eq:classical_uncertainty_1d}]
		Observe that under a change of variable induced by a $C^1$-diffeomorphism $\Phi:\LLRR^2\to\LLRR^2$, the Sobolev inequality gives for $p<2$ and $\frac{1}{q} = \frac{1}{p} - \frac{1}{2}$
		\begin{equation*}
			\LLNrm{\LLn{\det\nabla\Phi}^{1/q}f}{L^q(\LLRR^2)} \leq \LLcC^\textnormal{S}_{1,p} \LLNrm{\LLn{\det\nabla\Phi}^{1/p}(\nabla\Phi)^{-1}\nabla f}{L^p} .
		\end{equation*}
		In the case when $\Phi(z) = \left(\alpha(z),\beta(z)\right)$, notice that $\det\nabla\Phi = \{\alpha,\beta\}$, and $(\nabla\Phi)^{-1} \nabla g = \{\alpha,\beta\}^{-1}\left(-\{\beta,g\},\{\alpha,g\}\right)$. Therefore,
		\begin{equation*}
			\frac{1}{\LLcC^\textnormal{S}_{1,p}}\LLNrm{\LLn{\{\alpha,\beta\}}^\frac{1}{q}f}{L^q(\LLRR^2)} \leq \LLNrm{\LLn{\{\alpha,\beta\}}^{2/p}\{\alpha,\beta\}^{-2}\left(\LLn{\{\beta,f\}}^2+\LLn{\{\alpha,f\}}^2\right)}{L^{p/2}(\LLRR^2)}^{1/2} .
		\end{equation*}
		In particular, if $p = 1$ and $q = 2$, it gives
		\begin{equation*}
			\frac{1}{\LLcC^\textnormal{S}_{1,1}}\LLNrm{\LLn{\{\alpha,\beta\}}^\frac{1}{2}f}{L^2(\LLRR^2)}
			\\
			\leq \LLNrm{(\{\alpha,f\},\{\beta,f\})}{L^1(\LLRR^2)} .
		\end{equation*}
		Using the triangle inequality, changing $\alpha$ by $\lambda \alpha$ and $\beta$ by $\beta/\lambda$, and optimizing with respect to $\lambda$ gives
		\begin{equation*}
			\frac{1}{2\,(\LLcC^\textnormal{S}_{1,1})^2} \LLNrm{\LLn{\{\alpha,\beta\}}^\frac{1}{2}f}{L^2(\LLRR^2)}^2 \leq \LLNrm{\{\alpha,f\}}{L^1(\LLRR^2)} \LLNrm{\{\beta,f\}}{L^1(\LLRR^2)} .
		\end{equation*}
		Since $\{\alpha,f\} = \{\alpha,\sqrt{f}\}\sqrt{f}$, the result follows from the Cauchy--Schwarz inequality and the fact that $\LLNrm{f}{L^1(\LLRR^2)} = 1$.
	\end{proof}
	
\subsection{Bound for Commutators}\label{sec:quest_commutators}

	The quantum Sobolev spaces are useful to prove that commutators are small in terms of $\LLhb$. It is for example proved in~\cite{lafleche_quantum_2024} that for any $p\in[1,\infty]$
	\begin{equation*}
		\frac{1}{\sqrt{d}} \LLNrm{\LLDhv\LLop}{\LLL^p} \leq \sup_{\xi\in\LLRd\setminus\{0\}} \frac{\LLNrm{\com{e^{i\,\xi\cdot x}, \LLop}}{\LLL^p}}{\LLhb\LLn{\xi}} \leq \LLNrm{\LLDhv\LLop}{\LLL^p}.
	\end{equation*}
	It also follows for instance from \cite{potapov_operator-lipschitz_2011} that for any $p\in (1,\infty)$, there exists $c_p > 0$ such that for any $u \in W^{1,\infty}(\LLRd)$,
	\begin{equation}\label{eq:operator-lipschitz}
		\LLNrm{\com{u(x),\LLop}}{\LLL^p} \leq c_p\, \LLhb \LLNrm{\nabla u}{L^\infty} \LLNrm{\LLDhv\LLop}{\LLL^p}.
	\end{equation}
	By conjugation by the Fourier transform, which is a unitary operator, one obtains similarly
	\begin{equation*}
		\LLNrm{\com{u(\nabla),\LLop}}{\LLL^p} \leq c_p \LLNrm{\nabla u}{L^\infty} \LLNrm{\LLDhx\LLop}{\LLL^p}.
	\end{equation*}
	The cases $p\in\set{1,\infty}$ are however known to be false, one needs slightly more regularity for $u$ (see e.g.~\cite{aleksandrov_operator_2016}).
	
	These two inequalities are special cases of analogues to the classical H\"older's inequality, which gives for $1\leq p,q,r\leq \infty$ such that $\frac{1}{p} = \frac{1}{q} + \frac{1}{r}$ and for two functions $(\alpha,\beta) \in W^{1,q}(\LLRdd)\times W^{1,r}(\LLRdd)$ 
	\begin{equation*}
		\LLNrm{\{\alpha,\beta\}}{L^p} \leq \LLNrm{\nabla_x\alpha}{L^q} \LLNrm{\nabla_v\beta}{L^r} + \LLNrm{\nabla_v\alpha}{L^q} \LLNrm{\nabla_x\beta}{L^r} .
	\end{equation*}
	Replacing $\alpha$ by $\lambda\alpha$ and $\beta$ by $\lambda\beta$ and optimizing with respect to $\lambda$ gives
	\begin{equation*}
		\LLNrm{\{\alpha,\beta\}}{L^p}^2 \leq 2 \LLNrm{\nabla_x\alpha}{L^q} \LLNrm{\nabla_v\alpha}{L^q} \LLNrm{\nabla_v\beta}{L^r}  \LLNrm{\nabla_x\beta}{L^r}.
	\end{equation*}
	This implies in particular that $\LLNrm{\{\alpha,\beta\}}{L^p} \leq 2 \LLNrm{\nabla_z\alpha}{L^q} \LLNrm{\nabla_z\beta}{L^r}$, where $\nabla_z = (\nabla_x,\nabla_v)$. It would be useful to know whether the quantum generalization of such inequality holds, that is, if for $1<p,q,r<\infty$ such that $\frac{1}{p} = \frac{1}{q} + \frac{1}{r}$, there exists a constant $C$ independent of $\LLhb$ such that for any compact operators $A$ and $B$,
	\begin{equation*}
		\LLNrm{\com{A,B}}{\LLL^p} \leq C \,\LLhb \LLNrm{\LLDh A}{\LLL^q} \LLNrm{\LLDh B}{\LLL^r}.
	\end{equation*}
	This would generalize the ``commutator-Lipschitz" inequality~\eqref{eq:operator-lipschitz}, which is the case when $A = u(x)$ is a multiplication operator and $q=\infty$.
	
	One can already observe that the following weaker inequality holds for $n > d$
	\begin{equation}\label{eq:Holder_gradients}
		\LLNrm{\com{A,B}}{\LLL^p} \leq C_{d,n} \LLNrm{\LLDh^{n+1}\! A}{\LLL^2}^{d/n} \LLNrm{\LLDh A}{\LLL^2}^{1-d/n} \LLNrm{\LLDh B}{\LLL^p} .
	\end{equation}
	
	\begin{proof}[Proof of Inequality~\eqref{eq:Holder_gradients}]
		By the formula for the Weyl quantization, with $z=(y,v)$ and $e_z := e^{2i\pi \left(x\cdot y+v\cdot \LLopp\right)}$,
		\begin{equation*}
			A = \LLintdd \widehat{f}_A(z) \,e_z \LLd z \, .
		\end{equation*}
		Therefore by the triangle inequality
		\begin{equation*}
			\LLNrm{\com{A,B}}{\LLL^p} \leq \LLintdd \big|\widehat{\nabla f_A}(z)\big|\, \frac{\LLNrm{\com{e_z,B}}{\LLL^p}}{\LLn{2\pi z}} \LLd z \leq \LLbNrm{\widehat{f}_{\LLDh A}}{L^1(\LLRdd)} \sup_{z\neq 0} \frac{\LLNrm{\com{e_z,B}}{\LLL^p}}{\LLn{2\pi z}}.
		\end{equation*}
		By the same argument as in~\cite[Proof of Proposition~2.3]{lafleche_quantum_2024}, the last factor is bounded by $\LLNrm{\LLDh B}{\LLL^p}$. Therefore the result follows using the classical bound\footnote{The Sobolev space $H^n(\LLRdd)$ could also be replaced by the sharper Besov space $B^d_{2,1}(\LLRdd)$.} $\LLsNrm{\widehat{f}}{L^1(\LLRdd)} \leq C \LLNrm{f}{H^n(\LLRdd)}$, or equivalently, by a scaling argument,
		\begin{equation*}
			\LLsNrm{\widehat{f}}{L^1(\LLRdd)} \leq C \LLNrm{f}{L^2(\LLRdd)}^{1-d/n} \LLNrm{\nabla^n f}{L^2(\LLRdd)}^{d/n} ,
		\end{equation*}
		applied to $f = f_{\LLDh A}$, together with the fact that the Wigner transform is an isometry from $\LLL^2$ to $L^2(\LLRdd)$.
	\end{proof}

\end{document}